\documentclass[12pt]{article}
\PassOptionsToPackage{hyphens}{url}
\usepackage[pagebackref=true,colorlinks]{hyperref}
\usepackage{amsmath,amssymb,amsthm,mathtools,amsrefs,mathrsfs,amstext,amsxtra,latexsym,amsbsy,amsfonts,euscript}
\usepackage{bbm}
\usepackage{pifont}
\usepackage{mathpazo} 
\usepackage{eulervm}
\usepackage{csquotes}
\usepackage{lscape}
\usepackage{etex} 
\usepackage[table]{xcolor}
\usepackage{tikz}
\usetikzlibrary{arrows.meta}
\usetikzlibrary{decorations.pathmorphing,shapes}
\usetikzlibrary{calc}
\usepackage{comment}
\usepackage{adjustbox}
\usepackage{scalerel,stackengine}
\usepackage{stmaryrd}
\usepackage{fancyhdr}
\usepackage{soul}
\usepackage{dsfont}
\usepackage[normalem]{ulem}
\usepackage{longtable}
\usepackage[bottom]{footmisc}
\usepackage{graphicx}
\usepackage{caption}
\usepackage{subcaption}
\usepackage{pictexwd, dcpic}
\usepackage{float}
\usepackage{enumerate}
\usepackage[all]{xy} 
\usetikzlibrary{circuits.ee.IEC}
\hypersetup{
	colorlinks=true,
    linktoc=all,
    linkcolor=blue,  
    citecolor=blue,
    }

\usepackage{tocbasic}
\DeclareTOCStyleEntry[
  beforeskip=.2em plus 1pt,
  pagenumberformat=\textbf
]{tocline}{section}

\makeatletter
\newcommand*{\shifttext}[2]{%
  \settowidth{\@tempdima}{#2}%
  \makebox[\@tempdima]{\hspace*{#1}#2}%
}
\makeatother
\makeatletter
\renewcommand*\env@matrix[1][\arraystretch]{%
  \edef\arraystretch{#1}%
  \hskip -\arraycolsep
  \let\@ifnextchar\new@ifnextchar
  \array{*\c@MaxMatrixCols c}}
\makeatother

\stackMath
\newcommand\reallywidehat[1]{%
\savestack{\tmpbox}{\stretchto{%
  \scaleto{%
    \scalerel*[\widthof{\ensuremath{#1}}]{\kern.1pt\mathchar"0362\kern.1pt}%
    {\rule{0ex}{\textheight}}
  }{\textheight}%
}{2.4ex}}%
\stackon[-6.9pt]{#1}{\tmpbox}%
}
\makeatletter
\pgfmathdeclarefunction{erf}{1}{%
  \begingroup
    \pgfmathparse{#1 > 0 ? 1 : -1}%
    \edef\sign{\pgfmathresult}%
    \pgfmathparse{abs(#1)}%
    \edef\x{\pgfmathresult}%
    \pgfmathparse{1/(1+0.3275911*\x)}%
    \edef\t{\pgfmathresult}%
    \pgfmathparse{%
      1 - (((((1.061405429*\t -1.453152027)*\t) + 1.421413741)*\t 
      -0.284496736)*\t + 0.254829592)*\t*exp(-(\x*\x))}%
    \edef\y{\pgfmathresult}%
    \pgfmathparse{(\sign)*\y}%
    \pgfmath@smuggleone\pgfmathresult%
  \endgroup
}
\makeatother

\theoremstyle{theorem}
\newtheorem{theorem}[equation]{Theorem}
\newtheorem{lemma}[equation]{Lemma}
\newtheorem{proposition}[equation]{Proposition}
\newtheorem{corollary}[equation]{Corollary}
\theoremstyle{definition}
\newtheorem{definition}[equation]{Definition}
\newtheorem{construction}[equation]{Construction}
\newtheorem{question}[equation]{Question}

\newtheorem{problem}[equation]{Problem}
\newtheorem{example}[equation]{Example}
\newtheorem{exercise}[equation]{Exercise}
\newtheorem*{answer}{Answer}
\newtheorem*{solution}{Solution}
\newtheorem{remark}[equation]{Remark}

\newtheorem{notation}[equation]{Notation}
\newtheorem{noterm}[equation]{Notation and Terminology}

\newcommand\define[1]{\emph{\textbf{#1}}}

\numberwithin{equation}{section}

 \let\t=\tau

\newcommand{\be}{\begin{equation}}
\newcommand{\ee}{\end{equation}}
\def\ba{\begin{align}} 
\def\ea{\end{align}}
\newcommand{\bea}{\begin{eqnarray}}
\newcommand{\eea}{\end{eqnarray}}
\newcommand{\bx}{\begin{example}}
\newcommand{\ex}{\end{example}}
\newcommand{\bex}{\begin{exercise}}
\newcommand{\eex}{\end{exercise}}
\newcommand{\ban}{\begin{answer}}
\newcommand{\ean}{\end{answer}}
\newcommand{\bt}{\begin{theorem}}
\newcommand{\et}{\end{theorem}}
\newcommand{\bc}{\begin{corollary}}
\newcommand{\ec}{\end{corollary}}
\newcommand{\blem}{\begin{lemma}}
\newcommand{\elem}{\end{lemma}}
\newcommand{\bp}{\begin{problem}}
\newcommand{\ep}{\end{problem}}
\newcommand{\bn}{\begin{proposition}}
\newcommand{\en}{\end{proposition}}
\newcommand{\bd}{\begin{definition}}
\newcommand{\ed}{\end{definition}}
\newcommand{\bcon}{\begin{construction}}
\newcommand{\econ}{\end{construction}}
\newcommand{\bq}{\begin{question}}
\newcommand{\eq}{\end{question}}
\newcommand{\bprf}{\begin{proof}}
\newcommand{\eprf}{\end{proof}}
\newcommand{\br}{\begin{remark}}
\newcommand{\er}{\end{remark}}
\newcommand{\bs}{\begin{solution}}
\newcommand{\es}{\end{solution}}
\newcommand{\beqs}{\begin{eqnarray}}
\newcommand{\eeqs}{\end{eqnarray}}
\newcommand{\bnt}{\begin{noterm}}
\newcommand{\ent}{\end{noterm}}
\newcommand{\bnot}{\begin{notation}}
\newcommand{\enot}{\end{notation}}
\newcommand{\sto}{\rightsquigarrow}

\newcommand{\id}{\mathrm{id}}

\def\R{{{\mathbb R}}}

\newcommand{\FinStat}{\mathbf{FinStat}}

\newcommand{\stoch}{\;\xy0;/r.25pc/:(-3,0)*{}="1";(3,0)*{}="2";{\ar@{~>}"1";"2"|(1.06){\hole}};\endxy\!}

\newcounter{sarrow}
\newcommand\xstoch[1]{%
\stepcounter{sarrow}%
\mathrel{\begin{tikzpicture}[baseline= {( $ (current bounding box.south) + (0,-0.1ex) $ )}]
\node[inner sep=.5ex] (\thesarrow) {\;$\scriptstyle #1$\;};
\path[draw,{<[scale=1.5,width=3,length=2]}-,decorate,
  decoration={snake,amplitude=0.3mm,segment length=2.1mm,pre=lineto,pre length=1pt}] 
    (\thesarrow.south east) -- (\thesarrow.south west);
\end{tikzpicture}}%
}

\newcounter{sqarrow}

\newcommand{\ben}{\renewcommand{\theenumi}{\alph{enumi}} 
\renewcommand{\labelenumi}{(\theenumi)}\begin{enumerate}}
\newcommand{\een}{\end{enumerate}}
\newlength\stateheight
\newlength\minimumstatewidth
\tikzset{width/.initial=\minimummorphismwidth}
\tikzset{colour/.initial=white}

\newif\ifblack\pgfkeys{/tikz/black/.is if=black}
\newif\ifwedge\pgfkeys{/tikz/wedge/.is if=wedge}
\newif\ifvflip\pgfkeys{/tikz/vflip/.is if=vflip}
\newif\ifhflip\pgfkeys{/tikz/hflip/.is if=hflip}
\newif\ifhvflip\pgfkeys{/tikz/hvflip/.is if=hvflip}

\pgfdeclarelayer{foreground}
\pgfdeclarelayer{background}
\pgfsetlayers{background,main,foreground}

\def\thickness{0.4pt}

\makeatletter
\pgfkeys{%
  /tikz/on layer/.code={
    \pgfonlayer{#1}\begingroup
    \aftergroup\endpgfonlayer
    \aftergroup\endgroup
  },
  /tikz/node on layer/.code={
    \gdef\node@@on@layer{%
      \setbox\tikz@tempbox=\hbox\bgroup\pgfonlayer{#1}\unhbox\tikz@tempbox\endpgfonlayer\pgfsetlinewidth{\thickness}\egroup}
    \aftergroup\node@on@layer
  },
  /tikz/end node on layer/.code={
    \endpgfonlayer\endgroup\endgroup
  }
}
\def\node@on@layer{\aftergroup\node@@on@layer}
\makeatother

\makeatletter
\pgfdeclareshape{state}
{
    \savedanchor\centerpoint
    {
        \pgf@x=0pt
        \pgf@y=0pt
    }
    \anchor{center}{\centerpoint}
    \anchorborder{\centerpoint}
    \saveddimen\overallwidth
    {
        \pgf@x=3\wd\pgfnodeparttextbox
        \ifdim\pgf@x<\minimumstatewidth
            \pgf@x=\minimumstatewidth
        \fi
    }
    \savedanchor{\upperrightcorner}
    {
        \pgf@x=.5\wd\pgfnodeparttextbox
        \pgf@y=.5\ht\pgfnodeparttextbox
        \advance\pgf@y by -.5\dp\pgfnodeparttextbox
    }
    \anchor{A}
    {
        \pgf@x=-\overallwidth
        \divide\pgf@x by 4
        \pgf@y=0pt
    }
    \anchor{B}
    {
        \pgf@x=\overallwidth
        \divide\pgf@x by 4
        \pgf@y=0pt
    }
    \anchor{text}
    {
        \upperrightcorner
        \pgf@x=-\pgf@x
        \ifhflip
            \pgf@y=-\pgf@y
            \advance\pgf@y by 0.4\stateheight
        \else
            \pgf@y=-\pgf@y
            \advance\pgf@y by -0.4\stateheight
        \fi
    }
    \backgroundpath
    {
       \begin{pgfonlayer}{foreground}
        \pgfsetstrokecolor{black}
        \pgfsetlinewidth{\thickness}
        \pgfpathmoveto{\pgfpoint{-0.5*\overallwidth}{0}}
        \pgfpathlineto{\pgfpoint{0.5*\overallwidth}{0}}
        \ifhflip
            \pgfpathlineto{\pgfpoint{0}{\stateheight}}
        \else
            \pgfpathlineto{\pgfpoint{0}{-\stateheight}}
        \fi
        \pgfpathclose
        \ifblack
            \pgfsetfillcolor{black!50}
            \pgfusepath{fill,stroke}
        \else
            \pgfusepath{stroke}
        \fi
       \end{pgfonlayer}
    }
}

\pgfdeclareshape{my ground}
{
  \inheritsavedanchors[from=rectangle ee]
  \inheritanchor[from=rectangle ee]{center}
  \inheritanchor[from=rectangle ee]{north}
  \inheritanchor[from=rectangle ee]{south}
  \inheritanchor[from=rectangle ee]{east}
  \inheritanchor[from=rectangle ee]{west}
  \inheritanchor[from=rectangle ee]{north east}
  \inheritanchor[from=rectangle ee]{north west}
  \inheritanchor[from=rectangle ee]{south east}
  \inheritanchor[from=rectangle ee]{south west}
  \inheritanchor[from=rectangle ee]{input}
  \inheritanchor[from=rectangle ee]{output}
  \anchorborder{\csname pgf@anchor@my ground@input\endcsname}

  \backgroundpath{
    \pgf@process{\pgfpointadd{\southwest}{\pgfpoint{\pgfkeysvalueof{/pgf/outer xsep}}{\pgfkeysvalueof{/pgf/outer ysep}}}}
    \pgf@xa=\pgf@x \pgf@ya=\pgf@y
    \pgf@process{\pgfpointadd{\northeast}{\pgfpointscale{-1}{\pgfpoint{\pgfkeysvalueof{/pgf/outer xsep}}{\pgfkeysvalueof{/pgf/outer ysep}}}}}
    \pgf@xb=\pgf@x \pgf@yb=\pgf@y
    \pgfpathmoveto{\pgfqpoint{\pgf@xa}{\pgf@ya}}
    \pgfpathlineto{\pgfqpoint{\pgf@xa}{\pgf@yb}}
    \pgfmathsetlength\pgf@xa{.5\pgf@xa+.5\pgf@xb}
    \pgfmathsetlength\pgf@yc{.16666\pgf@yb-.16666\pgf@ya}
    \advance\pgf@ya by\pgf@yc
    \advance\pgf@yb by-\pgf@yc
    \pgfpathmoveto{\pgfqpoint{\pgf@xa}{\pgf@ya}}
    \pgfpathlineto{\pgfqpoint{\pgf@xa}{\pgf@yb}}
    \advance\pgf@ya by\pgf@yc
    \advance\pgf@yb by-\pgf@yc
    \pgfpathmoveto{\pgfqpoint{\pgf@xb}{\pgf@ya}}
    \pgfpathlineto{\pgfqpoint{\pgf@xb}{\pgf@yb}}
  }
}
\makeatother

\tikzset{inline text/.style =
  {text height=1.2ex,text depth=0.25ex,yshift=0.5mm}}
\tikzset{arrow box/.style =
  {rectangle,inline text,fill=white,draw,
    minimum height=5mm,yshift=-0.5mm,minimum width=5mm}}
\tikzset{bubble/.style =
  {inner sep=0mm,minimum width=3mm,minimum height=3mm,
    draw,shape=circle,fill=white}}
\tikzset{dot/.style =
  {inner sep=0mm,minimum width=1mm,minimum height=1mm,
    draw,shape=circle}}
\tikzset{white dot/.style = {dot,fill=white,text depth=-0.2mm}}
\tikzset{scalar/.style = {diamond,draw,inner sep=1pt}}
\tikzset{square/.style =
  {inner sep=0mm,minimum width=2mm,minimum height=2mm,
    draw,shape=rectangle}}

\tikzset{star/.style = {dot,fill=white,text depth=-0.2mm}}
\tikzset{copier/.style = {dot,fill,text depth=-0.2mm}}
\tikzset{fakecopier/.style = {square,fill,text depth=-0.2mm}}
\tikzset{discarder/.style = {my ground,draw,inner sep=0pt,
    minimum width=4.2pt,minimum height=11.2pt,anchor=input,rotate=90}}

\tikzset{xshiftu/.style = {shift = {(#1, 0)}}}
\tikzset{yshiftu/.style = {shift = {(0, #1)}}}
\tikzset{scriptstyle/.style={font=\everymath\expandafter{\the\everymath\scriptstyle}}}

\begin{document}

\begin{center}{\Large On a 2-relative entropy}\\
James Fullwood \end{center}

\abstract{We construct a 2-categorical extension of the relative entropy functor of Baez and Fritz, and show that our construction is functorial with respect to vertical morphisms. Moreover, we show such a `2-relative entropy' satisfies natural 2-categorial analogues of convex linearity, vanishing under optimal hypotheses, and lower semicontinuity. While relative entropy is a relative measure of information between probability distributions, we view our construction as a relative measure of information between \emph{channels}.}

\tableofcontents

\section{Introduction}

Let $X$ and $Y$ be finite sets which are the input and output alphabets of a discrete memoryless channel $X\xstoch{f} Y$ with probability transition matrix $f_{yx}$, representing the probability of the output $y$ given the input $x$. Every input $x$ then determines a probability distribution on $Y$ which we denote by $f^x$, so that $f^x(y)=f_{yx}$ for all $x\in X$ and $y\in Y$. The channel $X\xstoch{f} Y$ together with the choice of a prior distribution $p$ on $X$ will be denoted $(f|p)$, and such data then determines a distribution $\vartheta(f|p)$ on $X\times Y$ given by $\vartheta(f|p)_{(x,y)}=p_xf_{yx}$.  Given a second channel $X\xstoch{g} Y$ with prior distribution $q$ on $X$, the chain rule for relative entropy says that the relative entropy $D\left(\vartheta(f|p),\vartheta(g|q)\right)$ is given by
\be\label{RE1X}
D\left(\vartheta(f|p),\vartheta(g|q)\right)=D(p,q)+\sum_{x\in X}p_xD(f^x,g^x).
\ee
As the RHS of \eqref{RE1X} involves precisely the datum of the channels $f$ and $g$ together with the prior distributions $p$ and $q$, we view the quantity $D\left(\vartheta(f|p),\vartheta(g|q)\right)$ as a relative measure of information between the \emph{channels} $(f|p)$ and $(g|q)$. In particular, since from a Bayesian perspective $D(p,q)$ may be thought of as the amount of information gained upon discovering that the assumed prior distribution $p$ is actually $q$, it seems only natural to think of $D\left(\vartheta(f|p),\vartheta(g|q)\right)$ as the amount of information gained upon learning that the assumed channel $(f|p)$ is actually the channel $(g|q)$. 

To make such a Bayesian interpretation more precise, we build upon the work of Baez and Fritz \cite{BFRE}, who formulate a type of Bayesian inference as a \emph{process} $X\to Y$ (including a set of conditional hypotheses on the outcome of the process), which given a prior distribution $p$ on $X$ yields distributions $r$ on $Y$ and $q$ on $X$ in such a way that the relative entropy $D(p,q)$ has an operational meaning as a quantity associated with a Bayesian updating with respect to the process $X\to Y$ \footnote{Here $X$ may be thought of more generally as the set of possible states of some system to be measured, while $Y$ may be thought of as the possible outcomes of the measurement.}. Baez and Fritz then prove that up to a constant multiple, the map on such processes given by
\be\label{REQX1971}
(X\to Y)\mapsto \text{RE}(X\to Y):=D(p,q)
\ee
is the unique map satisfying the following axioms. 

\begin{enumerate}
\item\label{A1}
\underline{\bf{Functoriality}}: Given a composition of processes $X\to Y\to Z$, 
\[
\text{RE}(X\to Y\to Z)=\text{RE}(X\to Y)+\text{RE}(Y\to Z).
\]
\item\label{A2}
\underline{\bf{Convex Linearity}}: Given a collection of processes $U^x\to V^x$ indexed by the elements $x\in X$ of a finite probability space $(X,p)$, 
\[
\text{RE}\left(\sum_{x\in X}p_{x}(U^x\to V^x)\right)=\sum_{x\in X}p_{x}\text{RE}\left(U^x\to V^x\right).
\]
\item\label{A3} 
\underline{\bf{Vanishing Under Optimal Hypotheses}}: If the conditional hypotheses associated with a process $X\to Y$ are optimal, then 
\[
\text{RE}(X\to Y)=0.
\]
\item\label{A4}
\underline{\bf{Continuity}}: The map $(X\to Y)\mapsto \text{RE}(X\to Y)$ is lower semi-continuous.
\end{enumerate}

While Baez and Fritz facilitate their exposition using the language of category theory, knowing that a category consists of a class of objects together with a class of composable arrows (i.e. morphisms) between objects is all that is needed for an appreciation of their construction. From such a perspective, the aforementioned processes $X\to Y$ are morphisms in a category $\bf{FinStat}$, and the relative entropy assignment given by \eqref{REQX1971} is then a map from morphisms in $\bf{FinStat}$ to $[0,\infty]$. 

In what follows, we elevate the construction of Baez and Fritz to the level of 2-categories (or more precisely, \emph{double categories}), whose 2-morphisms may be viewed as certain processes between processes, or rather, processes which connect one channel to another. In particular, we construct a category $\bf{FinStat}_2$ which is a 2-level extension of the category $\bf{FinStat}$ introduced by Baez and Fritz, and define a relative entropy assignment $\text{RE}_2$ on 2-morphisms via the chain rule as given by \eqref{RE1X}. Moreover, we show that such a `2-relative entropy' satisfies the natural 2-level analogues of axioms \ref{A1}-\ref{A4} as satisfied by the relative entropy map RE.

As abstract as a relative entropy of processes between processes may seem, Shannon's Noisy Channel Coding Theorem---which is a cornerstone of information theory---is essentially a statement about transforming a noisy channel into a noiseless one via a sequence of codings and encodings. From such a viewpoint, information theory is fundamentally about processes (i.e., a sequence of codings and encodings), between processes (i.e., channels), and it is precisely this viewpoint with which we will proceed. Furthermore, there is a growing recent interest in axiomatic and categorical approaches to information theory \cite{BFL}\cite{CFS17}\cite{Fo12}\cite{Fr19}\cite{FAMI}\cite{FP1}\cite{leinster2019short}\cite{PaEntropy}, and the present work is a direct outgrowth of such activity.

\section{The Category $\bf{FinStat}$}
In this section, we introduce the first-level structure of interest, which is the catgeory $\FinStat$ introduced by Baez and Fritz \cite{BFRE}. Though we use the language of categories, knowing that a category consists of a class of composable arrows between a class of objects is sufficient for the comprehension of all categorical notions in this work.

\bd
Let $X$ and $Y$ be finite sets. A \define{discrete memoryless channel} (or simply \define{channel} for short) $X\xstoch{f} Y$ associates every $x\in X$ with a probability distribution $f^x$ on $Y$. In such as case, the sets $X$ and $Y$ are referred to as the \define{set of inputs} and \define{set of outputs} of the channel $f$ respectively, and $f^x(y)$ is the probability of receiving the output $y$ given the input $x$, which will be denoted by $f_{yx}$. 
\ed

\bd
If $X\xstoch{f} Y$ and $Y\xstoch{g} Z$ are channels, then the composition $X\xstoch{g\circ f} Z$ is given by
\[
g_{zx}=\sum_{x\in X}g_{zy}f_{yx}
\]
for all $x\in X$ and $y\in Y$.
\ed

\br
If $X\xstoch{f} Y$ is a channel such that for every $x\in X$ there exists a $y\in Y$ with $f_{yx}=1$, then such a $y$ is necessarily unique given $x$, and as such, $f$ may be identified with a function $f:X\to Y$. In such a case, we say that the channel $f$ is \define{pure} (or \define{deterministic}), and from here on we will not distinguish the difference between a pure channel and the associated function from its set of inputs to its set of outputs.
\er

\bd
If $\star$ denotes a set with a single element, then a channel $\star \xstoch{p} X$ is simply a probability distribution on $X$, and in such a case we will use $p_x$ to denote the probability of $x$ as given by $p$ for all $x\in X$. The pair $(X,p)$ is then referred to as a \define{finite probability space}.
\ed

\bnot
The datum of a channel $X\xstoch{f} Y$ together with a prior distribution $\star \xstoch{p} X$ on its set of inputs will be denoted $(f|p)$. 
\enot

\bd
Let $\FinStat$ denote the category whose objects are finite probability spaces, and whose morphisms $(X,p)\longrightarrow (Y,q)$ consist of the following data:
\begin{itemize}
\item
A function $f:X\to Y$ such that $f\circ p=q$.
\item 
A channel $Y\xstoch{s} X$ such that $f\circ s=\id_Y$. In other words, $Y\xstoch{s} X$  is a \define{stochastic section} of $f:X\to Y$.
\end{itemize}
A morphism in $\FinStat$ is then summarized by a diagram of the form
\be\label{FSMX17}
\xy0;/r.25pc/:
(0,7.5)*+{\star}="1";
(-15,-15)*+{X}="2";
(15,-15)*+{Y,}="0";
{\ar@{~>}"1";"2"_{p}};
{\ar@{~>}"1";"0"^{q}};
{\ar"2";"0"_{f}};
{\ar@{~>}@/_1.5pc/"0";"2"_{s}};
\endxy
\ee
and a composition of morphisms in $\FinStat$ is obtained via function composition and composition of stochastic sections. In such a case, it is straightforward to show that a composition of stochastic sections is a stochastic section, et cetera. The morphism corresponding to diagram \eqref{FSMX17} will often be denoted $(f,p,s)$.
\ed

\br
Note that in diagram \eqref{FSMX17}, a straight arrow is used for $X\overset{f}\longrightarrow Y$ as $f$ is a function, as opposed to a noisy channel.
\er

\br
The operational interpretation of diagram \eqref{FSMX17} is as follows. The set $X$ is thought of as the set of possible states of system, and $f:X\to Y$ is then thought of as a measurement process, so that $Y$ is then thought of as the set of possible states of some measuring apparatus. The stochastic section $Y\xstoch{s} X$ is then thought of as a set of hypotheses about the state of the system given a state of the measuring apparatus. In particular, $s_{xy}$ is thought of as the probability the system was in state $x$ given the state $y$ of the measuring apparatus.   
\er

\bd\label{OH19}
If the stochastic section $Y\xstoch{s} X$ in diagram \eqref{FSMX17} is such that $s\circ q=p$, then $s$ will be referred to as an \define{optimal hypothesis} for $(f|p)$.
\ed

\bd
Let $(X,p)$ be a finite probability space, and let $U^x\xstoch{\mu^x} V^x$ be a collection of channels with prior distributions $\star \xstoch{q^x} U^x$ indexed by $X$. The \define{convex combination} of $(\mu^x|q^x)$ with respect to $(X,p)$ is the channel 
\[
\left.\left(\coprod_{x\in X}U^x\xstoch{\bigoplus_{x\in X}\mu^x}\coprod_{x\in X}V^x\right|\bigoplus_{x\in X}p_xq^x\right),
\] 
where $\bigoplus_{x\in X}\mu^x$ is the channel given by
\[
\left(\bigoplus_{x\in X}\mu^x\right)_{vu}=\begin{cases} \mu^x_{vu} \quad \text{if} \hspace{2mm} (v,u)\in V^x\times U^x \hspace{2mm} \text{for some} \hspace{2mm} x\in X \\
0 \quad \hspace{0.41cm} \text{otherwise} \\
\end{cases},
\]
with prior distribution $\star \xstoch{\bigoplus_{x\in X}p_xq^x} \coprod U^x$ given by
\[
\left(\bigoplus_{x\in X}p_xq^x\right)_u=p_{x_u}q^{x_u}_{u}, 
\]
where $x_u$ is such that $u\in U^{x_u}$. Such a convex combination will be denoted $\bigoplus_{x\in X}p_x(\mu^x|q^x)$.
\ed

\section{The Baez and Fritz Characterization of Relative Entropy}
We now recall the Baez and Fritz characterization of relative entropy in $\FinStat$. 

\bd
Let $(X,p)$ be a finite probability space, and let 
\[
\left(U^x\overset{\mu^x}\longrightarrow V^x,\star \xstoch{q^x} U^x,V^x\xstoch{s^x}U^x\right)
\]
 be a collection of morphisms in $\FinStat$ indexed by $X$. The \define{convex combination} of $(\mu^x,q^x,s^x)$ with respect to $(X,p)$ is the morphism $\bigoplus_{x\in X}p_x(\mu^x,q^x,s^x)$ in $\FinStat$ corresponding to the diagram
\be
\xy0;/r.25pc/:
(0,10)*+{\star}="1";
(-20,-20)*+{\coprod_{x\in X}U^x}="2";
(20,-20)*+{\coprod_{x\in X}V^x,}="0";
{\ar@{~>}"1";"2"_{\bigoplus_{x\in X}p_xq^x}};
{\ar@{~>}"1";"0"^{\bigoplus_{x\in X}p_xr^x}};
{\ar"2";"0"_{\bigoplus_{x\in X}\mu^x}};
{\ar@{~>}@/_1.5pc/"0";"2"_{\bigoplus_{x\in X}s^x}};
\endxy
\ee
where $r^x=\mu^x\circ q^x$ for all $x\in X$.
\ed

\bd
Let $(f,p,s)$ be a morphism in $\FinStat$, and let $r=s\circ f\circ p$. The \define{relative entropy} of $(f,p,s)$ is the non-negative extended real number $\text{RE}(f,p,s)\in [0,\infty]$ given by
\[
\text{RE}(f,p,s)=D(p,r),
\]
where $D(p,r)=\sum_xp_x\log(p_x/r_x)$ is the relative entropy between the distributions $p$ and $r$ on $X$. 
\ed

\bd
Let $F:\FinStat\to [0,\infty]$ be a map from the morphisms in $\FinStat$ to the extended non-negative reals $[0,\infty]$. 
\itemize
\item
$F$ is said to be \define{functorial} if and only if for every composition $(g\circ f,p,s\circ t)$ of morphisms in $\FinStat$ we have
\[
F(g\circ f,p,s\circ t)=F(f,p,s)+F(g,f\circ p,t).
\]
\item
$F$ is said to be \define{convex linear} if and only if for every convex combination $\bigoplus_{x\in X}p_x(\mu^x,q^x,s^x)$ of morphisms in $\FinStat$ we have
\[
F\left(\bigoplus_{x\in X}p_x(\mu^x,q^x,s^x)\right)=\sum_{x\in X}p_xF(\mu^x,q^x,s^x).
\]
\item 
$F$ is said to be \define{vanishing under optimal hypotheses} if and only if for every morphism $(f,p,s)$ in $\FinStat$ with $s$ an optimal hypothesis we have
\[
F(f,p,s)=0.
\]
\item
$F$ is said to be \define{lower semicontinuous} if and only if for every sequence of morphisms $(f,p_n,s_n)$ in $\FinStat$ converging to a morphism $(f,p,s)$ we have
\[
F(f,p,s)\leq \liminf_{n\to \infty}F(f,p_n,s_n).
\]
\ed

\bt[The Baez and Fritz Characterization of Relative Entropy]\label{BFRE}
Let $\mathbb{S}$ be the collection of maps from the morphisms in $\FinStat$ to $[0,\infty]$ which are functorial, convex linear, vanishing under optimal hypotheses and lower semicontinuous. Then the following statements hold.
\begin{enumerate}[i.] 
\item
The relative entropy $\emph{RE}$ is an element of $\mathbb{S}$.
\item
If $F\in \mathbb{S}$, then $F=c\emph{RE}$ for some non-negative constant $c\in \R$.
\end{enumerate}
\et

\section{The Category $\FinStat_2$}
In this section, we introduce the second-level structure of interest, namely, the double category $\FinStat_2$, which is a 2-level extension of $\FinStat$.

\bd
Let $\bold{FinStat}_2$ denote the 2-category whose objects and 1-morphisms coincide with those of $\bold{FinStat}$, and whose 2-morphisms are constructed as follows. Given 1-morphisms
\[
(\mu,p,s)=\xy0;/r.25pc/:
(0,7.5)*+{\{\star\}}="1";
(-15,-15)*+{X}="2";
(15,-15)*+{X'}="0";
{\ar@{~>}"1";"2"_{p}};
{\ar@{~>}"1";"0"^{p'}};
{\ar"2";"0"_{\mu}};
{\ar@{~>}@/_1.5pc/"0";"2"_{s}};
\endxy
\quad
\text{and}
\quad
(\nu,q,t)=\xy0;/r.25pc/:
(0,7.5)*+{\{\star\}}="1";
(-15,-15)*+{Y}="2";
(15,-15)*+{Y'}="0";
{\ar@{~>}"1";"2"_{q}};
{\ar@{~>}"1";"0"^{q'}};
{\ar"2";"0"_{\nu}};
{\ar@{~>}@/_1.5pc/"0";"2"_{t}};
\endxy
,
\] 
a 2-morphism  $\spadesuit:(\mu,p,s) \Rightarrow (\nu,q,t)$ consists of channels $f:X\sto Y$ and $f':X'\sto Y'$ such that 
\begin{itemize}
\item $f\circ p=q$ 
\item $f'\circ p'=q'$ 
\item $\nu\circ f=f'\circ \mu$ 
\end{itemize}
The two morphism $\spadesuit:(\mu,p,s) \Rightarrow (\nu,q,t)$ may then be summarized by the following diagram.
\be\label{PYRA17}
\underline{\spadesuit:(\mu,p,s) \Rightarrow (\nu,q,t)}
\xy0;/r.25pc/:
(0,7.5)*+{\star}="1";
(-12.5,-7.5)*+{X}="2";
(12.5,-7.5)*+{Y}="0";
(-12.5,-22.5)*+{X'}="3";
(12.5,-22.5)*+{Y'}="4";
(0,-37.5)*+{\star}="5";
{\ar@{~>}"1";"2"_{p}};
{\ar@{~>}"1";"0"^{q}};
{\ar@{~>}"2";"0"_{f}};
{\ar@{~>}"3";"4"_{f'}};
{\ar"2";"3"^{\mu}};
{\ar"0";"4"_{\nu}};
{\ar@{~>}@/_1.5pc/"4";"0"_{t}};
{\ar@{~>}@/^1.5pc/"3";"2"^{s}};
{\ar@{~>}"5";"3"^{p'}};
{\ar@{~>}"5";"4"_{q'}};
\endxy
\ee
\ed

\br
A crucial point is that in the above diagram, all arrows necessarily commute except for any compositions involving the the outer ``wings'', $s$ and $t$. For example, the compositions $s\circ \mu \circ p$ and $t\circ f'\circ \mu$ need not be equal to $p$ and $f$ respectively.
\er

\br
Diagram \eqref{PYRA17} should be thought of as a flattened out pyramid, whose base is the inner square and whose vertex is obtained by the identification of the upper and lower stars in the diagram.
\er

For vertical composition of 2-morphisms, suppose $\clubsuit:(\mu',p',s') \Rightarrow (\nu',q',t')$ is the 2-morphism summarized by the following diagram.
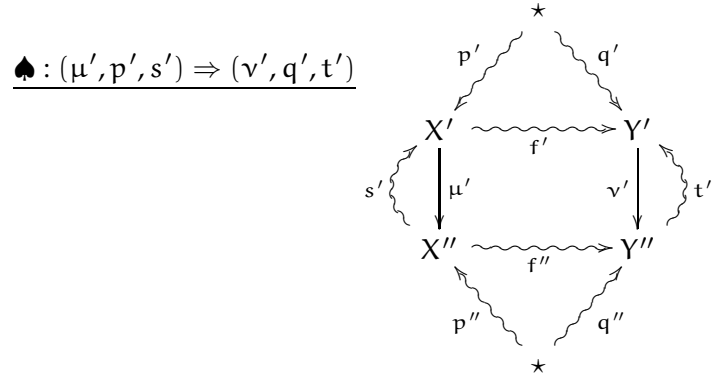
\[
\underline{\spadesuit:(\mu',p',s') \Rightarrow (\nu',q',t')}
\xy0;/r.25pc/:
(0,7.5)*+{\star}="1";
(-12.5,-7.5)*+{X'}="2";
(12.5,-7.5)*+{Y'}="0";
(-12.5,-22.5)*+{X''}="3";
(12.5,-22.5)*+{Y''}="4";
(0,-37.5)*+{\star}="5";
{\ar@{~>}"1";"2"_{p'}};
{\ar@{~>}"1";"0"^{q'}};
{\ar@{~>}"2";"0"_{f'}};
{\ar@{~>}"3";"4"_{f''}};
{\ar"2";"3"^{\mu'}};
{\ar"0";"4"_{\nu'}};
{\ar@{~>}@/_1.5pc/"4";"0"_{t'}};
{\ar@{~>}@/^1.5pc/"3";"2"^{s'}};
{\ar@{~>}"5";"3"^{p''}};
{\ar@{~>}"5";"4"_{q''}};
\endxy
\]
The vertical composition $\clubsuit\circ \spadesuit$ is then summarized by the following diagram.
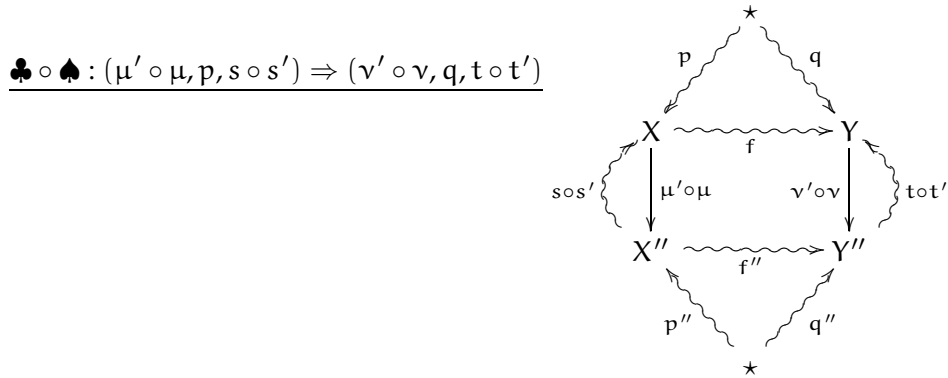
\[
\underline{\clubsuit\circ \spadesuit:(\mu'\circ\mu,p,s\circ s')\Rightarrow (\nu'\circ \nu,q,t\circ t')}
\xy0;/r.25pc/:
(0,7.5)*+{\star}="1";
(-12.5,-7.5)*+{X}="2";
(12.5,-7.5)*+{Y}="0";
(-12.5,-22.5)*+{X''}="3";
(12.5,-22.5)*+{Y''}="4";
(0,-37.5)*+{\star}="5";
{\ar@{~>}"1";"2"_{p}};
{\ar@{~>}"1";"0"^{q}};
{\ar@{~>}"2";"0"_{f}};
{\ar@{~>}"3";"4"_{f''}};
{\ar"2";"3"^{\mu'\circ \mu}};
{\ar"0";"4"_{\nu'\circ \nu}};
{\ar@{~>}@/_1.5pc/"4";"0"_{t\circ t'}};
{\ar@{~>}@/^1.5pc/"3";"2"^{s\circ s'}};
{\ar@{~>}"5";"3"^{p''}};
{\ar@{~>}"5";"4"_{q''}};
\endxy
\]

For horizontal composition, let $\heartsuit:(\nu,q,t)\Rightarrow (\xi,r,u)$ be a two morphism summarized by the following diagram.
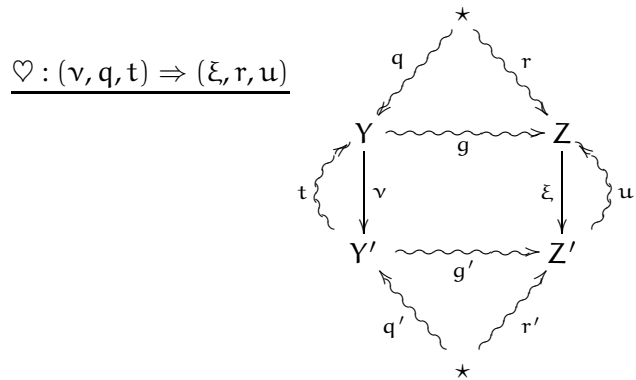
\[
\underline{\heartsuit:(\nu,q,t)\Rightarrow (\xi,r,u)}
\xy0;/r.25pc/:
(0,7.5)*+{\star}="1";
(-12.5,-7.5)*+{Y}="2";
(12.5,-7.5)*+{Z}="0";
(-12.5,-22.5)*+{Y'}="3";
(12.5,-22.5)*+{Z'}="4";
(0,-37.5)*+{\star}="5";
{\ar@{~>}"1";"2"_{q}};
{\ar@{~>}"1";"0"^{r}};
{\ar@{~>}"2";"0"_{g}};
{\ar@{~>}"3";"4"_{g'}};
{\ar"2";"3"^{\nu}};
{\ar"0";"4"_{\xi}};
{\ar@{~>}@/_1.5pc/"4";"0"_{u}};
{\ar@{~>}@/^1.5pc/"3";"2"^{t}};
{\ar@{~>}"5";"3"^{q'}};
{\ar@{~>}"5";"4"_{r'}};
\endxy
\]
The horizontal composition $\heartsuit\circ \spadesuit$ is then summarized by the following diagram.
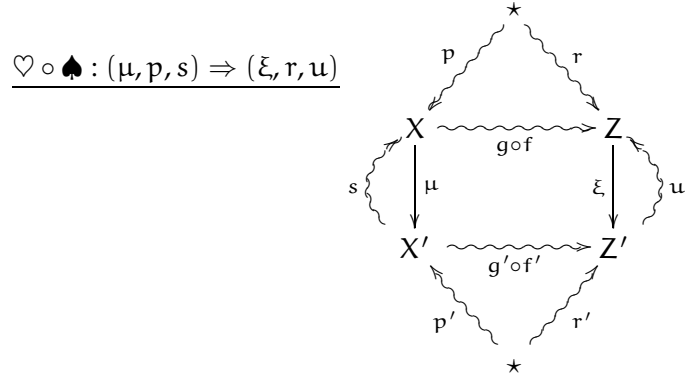
\[
\underline{\heartsuit\circ \spadesuit:(\mu,p,s)\Rightarrow (\xi,r,u)}
\xy0;/r.25pc/:
(0,7.5)*+{\star}="1";
(-12.5,-7.5)*+{X}="2";
(12.5,-7.5)*+{Z}="0";
(-12.5,-22.5)*+{X'}="3";
(12.5,-22.5)*+{Z'}="4";
(0,-37.5)*+{\star}="5";
{\ar@{~>}"1";"2"_{p}};
{\ar@{~>}"1";"0"^{r}};
{\ar@{~>}"2";"0"_{g\circ f}};
{\ar@{~>}"3";"4"_{g'\circ f'}};
{\ar"2";"3"^{\mu}};
{\ar"0";"4"_{\xi}};
{\ar@{~>}@/_1.5pc/"4";"0"_{u}};
{\ar@{~>}@/^1.5pc/"3";"2"^{s}};
{\ar@{~>}"5";"3"^{p'}};
{\ar@{~>}"5";"4"_{r'}};
\endxy
\]

\section{Convexity in $\bold{FinStat}_2$}
We now generalize the convex structure on morphisms in $\FinStat$ to 2-morphisms in $\FinStat_2$. For this, let $(X,p)$ be a finite probability space, and let $\spadesuit^{x}$ be a collection of 2-morphisms in $\bold{FinStat}_2$ indexed by $X$, where $\spadesuit^{x}$ is summarized by the following diagram.
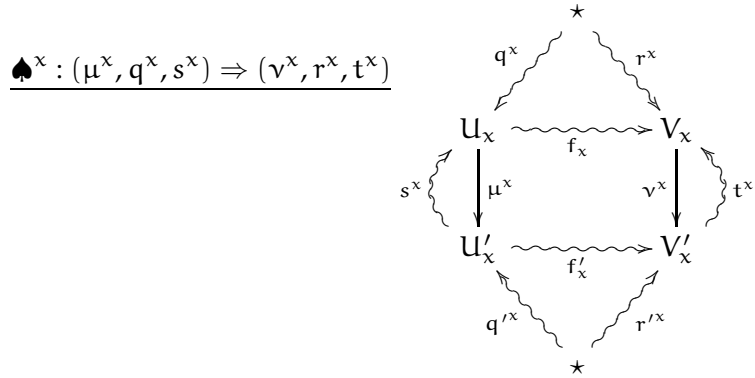
\[
\underline{\spadesuit^{x}:(\mu^x,q^x,s^x) \Rightarrow (\nu^x,r^x,t^x)}
\xy0;/r.25pc/:
(0,7.5)*+{\star}="1";
(-12.5,-7.5)*+{U_x}="2";
(12.5,-7.5)*+{V_x}="0";
(-12.5,-22.5)*+{U'_x}="3";
(12.5,-22.5)*+{V'_x}="4";
(0,-37.5)*+{\star}="5";
{\ar@{~>}"1";"2"_{q^x}};
{\ar@{~>}"1";"0"^{r^x}};
{\ar@{~>}"2";"0"_{f_x}};
{\ar@{~>}"3";"4"_{f'_x}};
{\ar"2";"3"^{\mu^x}};
{\ar"0";"4"_{\nu^x}};
{\ar@{~>}@/_1.5pc/"4";"0"_{t^x}};
{\ar@{~>}@/^1.5pc/"3";"2"^{s^x}};
{\ar@{~>}"5";"3"^{{q'}^{x}}};
{\ar@{~>}"5";"4"_{{r'}^{x}}};
\endxy
\]

\bd
The \define{convex sum} $\bigoplus_{x\in X}p_x\spadesuit^{x}$ is the 2-morphism in $\bold{FinStat}_2$ summarized by the following diagram.
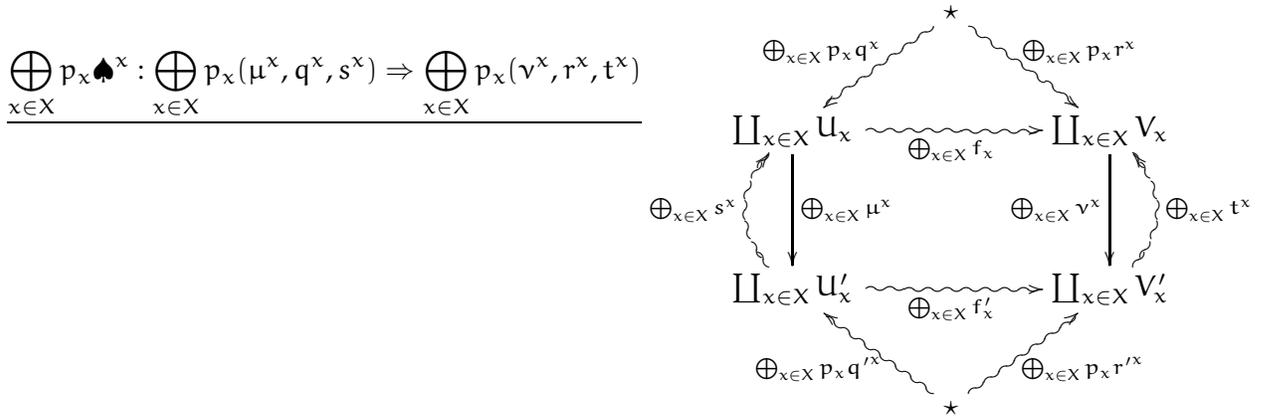
\[
\underline{\bigoplus_{x\in X}p_x\spadesuit^{x}:\bigoplus_{x\in X}p_x(\mu^x,q^x,s^x)\Rightarrow \bigoplus_{x\in X}p_x(\nu^x,r^x,t^x)}
\xy0;/r.25pc/:
(0,7.5)*+{\star}="1";
(-20,-7.5)*+{\coprod_{x\in X}U_x}="2";
(20,-7.5)*+{\coprod_{x\in X}V_x}="0";
(-20,-27.5)*+{\coprod_{x\in X}U'_x}="3";
(20,-27.5)*+{\coprod_{x\in X}V'_x}="4";
(0,-42.5)*+{\star}="5";
{\ar@{~>}"1";"2"_{\bigoplus_{x\in X} p_xq^x}};
{\ar@{~>}"1";"0"^{\bigoplus_{x\in X} p_xr^x}};
{\ar@{~>}"2";"0"_{\bigoplus_{x\in X} f_x}};
{\ar@{~>}"3";"4"_{\bigoplus_{x\in X} f'_x}};
{\ar"2";"3"^{\bigoplus_{x\in X} \mu^x}};
{\ar"0";"4"_{\bigoplus_{x\in X} \nu^x}};
{\ar@{~>}@/_1.5pc/"4";"0"_{\bigoplus_{x\in X} t^x}};
{\ar@{~>}@/^1.5pc/"3";"2"^{\bigoplus_{x\in X} s^x}};
{\ar@{~>}"5";"3"^{{\bigoplus_{x\in X} p_x{q'}^x}}};
{\ar@{~>}"5";"4"_{{\bigoplus_{x\in X} p_x{r'}^x}}};
\endxy
\]
\ed

\section{Conditional relative entropy in $\bold{FinStat}_2$}
We now introduce a measure of information associated with 2-morphisms in $\FinStat_2$  which we refer to as `conditional relative entropy'. The results proved in this section are essentially all lemmas for the results proved in the next section, where we introduce a 2-level extension of the relative entropy map RE, and show that it satisfies the 2-level analogues of the characterizing axioms of relative entropy.

\bd
With every 2-morphism 
\[
\underline{\spadesuit:(\mu,s|p) \Rightarrow (\nu,t|q)}
\xy0;/r.25pc/:
(0,7.5)*+{\star}="1";
(-12.5,-7.5)*+{X}="2";
(12.5,-7.5)*+{Y}="0";
(-12.5,-22.5)*+{X'}="3";
(12.5,-22.5)*+{Y'}="4";
(0,-37.5)*+{\star}="5";
{\ar@{~>}"1";"2"_{p}};
{\ar@{~>}"1";"0"^{q}};
{\ar@{~>}"2";"0"_{f}};
{\ar@{~>}"3";"4"_{f'}};
{\ar"2";"3"^{\mu}};
{\ar"0";"4"_{\nu}};
{\ar@{~>}@/_1.5pc/"4";"0"_{t}};
{\ar@{~>}@/^1.5pc/"3";"2"^{s}};
{\ar@{~>}"5";"3"^{p'}};
{\ar@{~>}"5";"4"_{q'}};
\endxy
\in \text{Mor}_2(\bold{FinStat}_2)
\]
we associate the non-negative extended real number $\text{CE}(\spadesuit)\in [0,\infty]$ given by
\be\label{CEDEF91}
\text{CE}(\spadesuit)=\sum_{x\in X}p_xD(f^x,(t\circ f'\circ \mu)^x),
\ee
where $D(-,-)$ is the standard relative entropy. We refer to $\text{CE}(\spadesuit)$ as the \define{conditional relative entropy} of $\spadesuit$.
\ed

\br
We refer to $\text{CE}(\spadesuit)$ as conditional relative entropy as its defining formula \eqref{CEDEF91} is structurally similar to the defining formula for conditional entropy. In particular, if $X\xstoch{f} Y$ is a channel with prior distribution $\star \xstoch{p} X$, then the conditional entropy $H(f|p)$ is given by
\[
H(f|p)=\sum_{x\in X}p_xH(f^x),
\]
where $H(f^x)$ is the Shannon entropy of the distribution $f^x$ on $Y$.
\er

\begin{proposition}\label{CECL979}
Conditional relative entropy in $\bold{FinStat}_2$ is convex linear, i.e., if $(X,p)$ is a finite probability space and $\spadesuit^{x}$ is a collection of 2-morphisms in $\bold{FinStat}_2$ indexed by $X$, then
\[
\emph{CE}\left(\bigoplus_{x\in X}p_x\spadesuit^x\right)=\sum_{x\in X}p_x \emph{CE}(\spadesuit^x).
\]
\end{proposition}
\begin{proof}
Suppose $\bigoplus_{x\in X}p_x\spadesuit^x$ is summarized by the following diagram
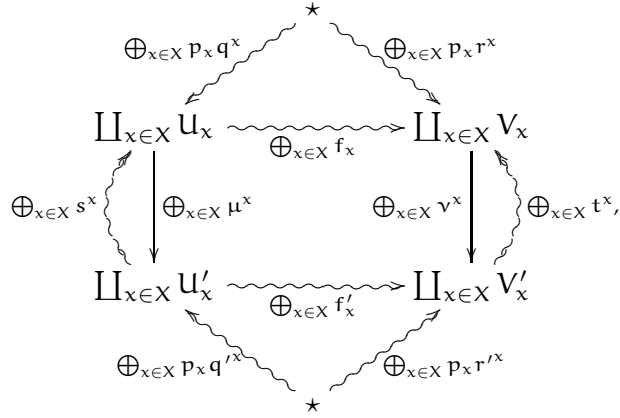
\[
\xy0;/r.25pc/:
(0,7.5)*+{\star}="1";
(-20,-7.5)*+{\coprod_{x\in X}U_x}="2";
(20,-7.5)*+{\coprod_{x\in X}V_x}="0";
(-20,-27.5)*+{\coprod_{x\in X}U'_x}="3";
(20,-27.5)*+{\coprod_{x\in X}V'_x}="4";
(0,-42.5)*+{\star}="5";
{\ar@{~>}"1";"2"_{\bigoplus_{x\in X} p_xq^x}};
{\ar@{~>}"1";"0"^{\bigoplus_{x\in X} p_xr^x}};
{\ar@{~>}"2";"0"_{\bigoplus_{x\in X} f_x}};
{\ar@{~>}"3";"4"_{\bigoplus_{x\in X} f'_x}};
{\ar"2";"3"^{\bigoplus_{x\in X} \mu^x}};
{\ar"0";"4"_{\bigoplus_{x\in X} \nu^x}};
{\ar@{~>}@/_1.5pc/"4";"0"_{\bigoplus_{x\in X} t^x,}};
{\ar@{~>}@/^1.5pc/"3";"2"^{\bigoplus_{x\in X} s^x}};
{\ar@{~>}"5";"3"^{{\bigoplus_{x\in X} p_x{q'}^x}}};
{\ar@{~>}"5";"4"_{{\bigoplus_{x\in X} p_x{r'}^x}}};
\endxy
\]
and let $U=\coprod_{x\in X}U_x$, $\mu=\oplus_{x\in X}\mu^x$, $f=\oplus_{x\in X}f^x$, $f'=\oplus_{x\in X}f'^x$ and $t=\oplus_{x\in X}t^x$. We then have
\begin{eqnarray*}
\text{CE}\left(\bigoplus_{x\in X}p_x\spadesuit^x\right)&=&\sum_{u\in U}q_u D(f^{u}, (t\circ f'\circ \mu)^u) \\
&=&\sum_{x\in X}\sum_{u_x\in U_x} p_xq^{x}_{u_x} D(f^{u_x}, (t\circ f'\circ \mu)^{u_x}) \\
&=&\sum_{x\in X}p_x\sum_{u_x\in U_x} q^{x}_{u_x} D((f_x)^{u_x}, (t^{x}\circ f'_{x}\circ \mu^{x})^{u_x}) \\
&=&\sum_{x\in X}p_x\text{CE}(\spadesuit^x), \\
\end{eqnarray*}
as desired.
\end{proof}

\begin{theorem}\label{CEFV19}
Conditional relative entropy in $\bold{FinStat}_2$ is functorial with respect to vertical composition, i.e., if $\clubsuit\circ \spadesuit$ is a vertical composition in $\bold{FinStat}_2$, then $\emph{CE}(\clubsuit\circ \spadesuit)=\emph{CE}(\clubsuit)+\emph{CE}(\spadesuit)$.
\end{theorem}

\begin{lemma}\label{LCM17}
Let $X\overset{f}\sto Y\overset{g}\sto Z$ be a composition of channels.
\begin{enumerate}
\item\label{T1}
 If $f$ is a pure channel, then $(g\circ f)_{zx}=g_{zf(x)}$.
\item\label{T2}
If $g$ is a stochastic section of a pure channel $Z\overset{h}\to Y$, then $(g\circ f)_{zx}=g_{zh(z)}f_{h(z)x}$.
\end{enumerate}
\end{lemma}

\bprf
The statements \ref{T1} and \ref{T2} follow immediately from the definitions of pure channel and stochastic section. 
\eprf

\begin{lemma}\label{LEV77}
Let $\spadesuit$ be a 2-morphism in $\bold{FinStat}_2$ as summarized by the diagram 
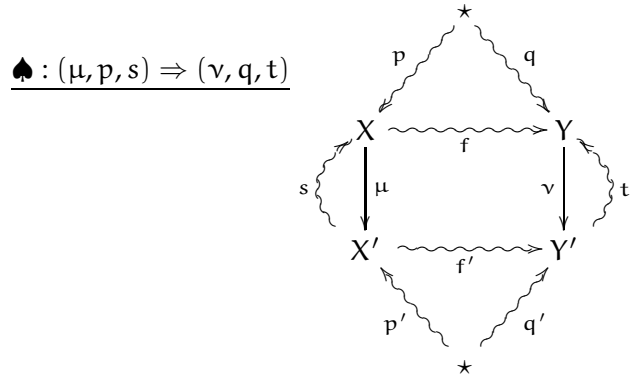
\[
\underline{\spadesuit:(\mu,p,s) \Rightarrow (\nu,q,t)}
\xy0;/r.25pc/:
(0,7.5)*+{\star}="1";
(-12.5,-7.5)*+{X}="2";
(12.5,-7.5)*+{Y}="0";
(-12.5,-22.5)*+{X'}="3";
(12.5,-22.5)*+{Y'}="4";
(0,-37.5)*+{\star}="5";
{\ar@{~>}"1";"2"_{p}};
{\ar@{~>}"1";"0"^{q}};
{\ar@{~>}"2";"0"_{f}};
{\ar@{~>}"3";"4"_{f'}};
{\ar"2";"3"^{\mu}};
{\ar"0";"4"_{\nu}};
{\ar@{~>}@/_1.5pc/"4";"0"_{t}};
{\ar@{~>}@/^1.5pc/"3";"2"^{s}};
{\ar@{~>}"5";"3"^{p'}};
{\ar@{~>}"5";"4"_{q'}};
\endxy
\]
Then 
\begin{enumerate}
\item $\emph{CE}(\spadesuit)=\sum_{x\in X}\sum_{y\in Y}p_xf_{yx}\log\left(\frac{f_{yx}}{t_{y\nu(y)}f'_{\nu(y)\mu(x)}}\right)$ for all $x\in X$. \\
\item $p'_{x'}f'_{y'x'}=\sum_{x\in \mu^{-1}(x')}\sum_{y\in \nu^{-1}(y')}p_xf_{yx}$ for all $(x',y')\in X'\times Y'$. \\
\end{enumerate}
\end{lemma}
\begin{proof} To prove item (1), let $x\in X$ and $y\in Y$. Then
\begin{eqnarray*}
(t\circ f'\circ \mu)_{yx}&=&(t\circ (f'\circ \mu))_{yx} \\
&=&\sum_{y'\in Y'}t_{yy'}(f'\circ \mu)_{y'x} \\
&=&\sum_{y'\in Y}t_{yy'}f'_{y'\mu(x)} \\
&=&t_{y\nu(y)}f'_{\nu(y)\mu(x)},
\end{eqnarray*}
where the third equality follows from Lemma~\ref{LCM17} since $\mu$ is a pure channel, and the fourth equality follows also from Lemma~\ref{LCM17} since $t$ is a stochastic section of a pure channel. We then have
\begin{eqnarray*}
\text{CE}(\spadesuit)&=&\sum_{x\in X}p_xD(f^x,(t\circ f'\circ \mu)^x) \\
&=&\sum_{x\in X}\sum_{y\in Y}p_xf_{yx}\log\left(\frac{f_{yx}}{(t\circ f'\circ \mu)_{yx}}\right) \\
&=&\sum_{x\in X}\sum_{y\in Y}p_xf_{yx}\log\left(\frac{f_{yx}}{t_{y\nu(y)}f'_{\nu(y)\mu(x)}}\right),
\end{eqnarray*}
as desired.

To prove item (2), the condition $\nu\circ f=f'\circ \mu$ is equivalent to the equation $(\nu\circ f)_{y'x}=(f'\circ \mu)_{y'x}$ for all $y'\in Y'$ and $x\in X$. And since
\[
(f'\circ \mu)_{y'x}=f'_{y'\mu(x)}, 
\]
and 
\[
(\nu\circ f)_{y'x}=\sum_{y\in Y}\nu_{y'y}f_{yx}=\sum_{y\in \nu^{-1}(y')}\nu_{y'y}f_{yx}=\sum_{y\in \nu^{-1}(y')}f_{yx},
\]
it follows that $f'_{y'\mu(x)}=\sum_{y\in \nu^{-1}(y')}f_{yx}$, thus for all $x'\in X'$ and $y'\in Y$ it follows that $f'_{y'x'}=\sum_{y\in \nu^{-1}(y')}f_{yx}$ for all $x\in \mu^{-1}(x')$. As such, we have
\[
p'_{x'}f'_{y'x'}=\left(\sum_{x\in \mu^{-1}(x')}p_x\right)\left(\sum_{y\in \nu^{-1}(y')}f_{yx}\right)=\sum_{x\in \mu^{-1}(x')}\sum_{y\in \nu^{-1}(y')}p_xf_{yx},
\]
as desired.
\end{proof}
\begin{proof}[Proof of Theorem~\ref{CEFV19}] 
Suppose $\spadesuit$ and $\clubsuit$ are such that the vertical composition $\clubsuit\circ \spadesuit$ is summarized by the following diagram.
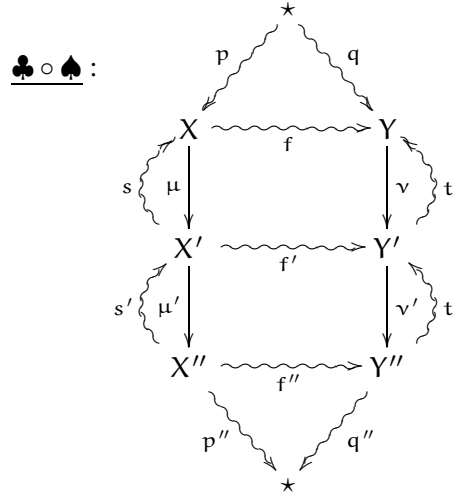
\[
\underline{\clubsuit\circ \spadesuit}:
\xy0;/r.25pc/:
(0,7.5)*+{\star}="1";
(-12.5,-7.5)*+{X}="2";
(12.5,-7.5)*+{Y}="0";
(-12.5,-22.5)*+{X'}="3";
(12.5,-22.5)*+{Y'}="4";
(-12.5,-37.5)*+{X''}="5";
(12.5,-37.5)*+{Y''}="6";
(0,-52.5)*+{\star}="7";
{\ar@{~>}"1";"2"_{p}};
{\ar@{~>}"1";"0"^{q}};
{\ar@{~>}"2";"0"_{f}};
{\ar@{~>}"3";"4"_{f'}};
{\ar"2";"3"_{\mu}};
{\ar"0";"4"^{\nu}};
{\ar"3";"5"_{\mu'}};
{\ar"4";"6"^{\nu'}};
{\ar@{~>}"5";"7"_{p''}};
{\ar@{~>}"6";"7"^{q''}};
{\ar@{~>}"5";"6"_{f''}};
{\ar@{~>}@/_1.5pc/"4";"0"_{t}};
{\ar@{~>}@/^1.5pc/"3";"2"^{s}};
{\ar@{~>}@/_1.5pc/"6";"4"_{t'}};
{\ar@{~>}@/^1.5pc/"5";"3"^{s'}};
\endxy
\]
By item (1) in Lemma~\ref{LEV77} we then have
\[
\text{CE}(\clubsuit\circ \spadesuit)=\sum_{x\in X}\sum_{y\in Y}p_xf_{yx}\log\left(\frac{f_{yx}}{(t\circ t')_{y(\nu'\circ \nu)(y)}f''_{(\nu'\circ \nu)(y)(\mu'\circ \mu)(x)}}\right), \\
\]
and since $t$ is a section of the pure channel $\nu$, by item (2) of Lemma~\ref{LCM17} it follows that for every $y\in Y$ we have
\[
(t\circ t')_{y(\nu'\circ \nu)(y)}=t_{y\nu(y)}t'_{\nu(y)(\nu'\circ \nu)(y)}.
\]
As such, we have
\begin{eqnarray*}
\text{CE}(\clubsuit\circ \spadesuit)&=&\sum_{x\in X}\sum_{y\in Y}p_xf_{yx}\log\left(\frac{f_{yx}}{t_{y\nu(y)}t'_{\nu(y)(\nu'\circ \nu)(y)}f''_{(\nu'\circ \nu)(y)}}\right) \\
&=&\sum_{x\in X}\sum_{y\in Y}p_xf_{yx}\log\left(\frac{f_{yx}f'_{\nu(y)\mu(x)}}{t_{y\nu(y)}f'_{\nu(y)\mu(x)}t'_{\nu(y)(\nu'\circ \nu)(y)}f''_{(\nu'\circ \nu)(y)(\mu'\circ \mu)(x)}}\right) \\
&=&\sum_{x\in X}\sum_{y\in Y}p_xf_{yx}\log\left(\frac{f_{yx}}{t_{y\nu(y)}f'_{\nu(y)\mu(x)}}\right) \\
&+& \sum_{x\in X}\sum_{y\in Y}p_xf_{yx}\log\left(\frac{f'_{\nu(y)\mu(x)}}{t'_{\nu(y)(\nu'\circ \nu)(y)}f''_{(\nu'\circ \nu)(y)(\mu'\circ \mu)(x)}}\right) \\
&=&\sum_{x\in X}\sum_{y\in Y}p_xf_{yx}\log\left(\frac{f_{yx}}{t_{y\nu(y)}f'_{\nu(y)\mu(x)}}\right) \\
&+&\sum_{x'\in X'}\sum_{y'\in Y'}p'_{x'}f'_{y'x'}\log\left(\frac{f'_{y'x'}}{t'_{y'\nu'(y')}f''_{\nu'(y')\mu'(x')}}\right) \quad  \quad (\text{by item (2) of \text{Lemma}~\ref{LEV77}}) \\
&=&\sum_{x\in X}p_xD(f^x,(t\circ f'\circ \mu)^x)+\sum_{x'\in X'}p'_{x'}D(f'^{x'},(t'\circ f''\circ \mu')^{x'}) \\
&=&\text{CE}(\clubsuit)+\text{CE}(\spadesuit),
\end{eqnarray*}
as desired (the second-to-last equality follows by item (1) of \text{Lemma}~\ref{LEV77}).
\end{proof}

\section{Relative entropy in $\bold{FinStat}_2$}
In this section we introduce a 2-level extension of the relative entropy map RE introduced by Baez and Fritz, and show that it satisfies the natural 2-level analogues of functoriality, convex linearity, vanishing under optimal hypotheses and lower semicontinuity.

\bd
With every 2-morphism 
\[
\underline{\spadesuit:(\mu,p,s) \Rightarrow (\nu,q,t)}
\xy0;/r.25pc/:
(0,7.5)*+{\star}="1";
(-12.5,-7.5)*+{X}="2";
(12.5,-7.5)*+{Y}="0";
(-12.5,-22.5)*+{X'}="3";
(12.5,-22.5)*+{Y'}="4";
(0,-37.5)*+{\star}="5";
{\ar@{~>}"1";"2"_{p}};
{\ar@{~>}"1";"0"^{q}};
{\ar@{~>}"2";"0"_{f}};
{\ar@{~>}"3";"4"_{f'}};
{\ar"2";"3"^{\mu}};
{\ar"0";"4"_{\nu}};
{\ar@{~>}@/_1.5pc/"4";"0"_{t}};
{\ar@{~>}@/^1.5pc/"3";"2"^{s}};
{\ar@{~>}"5";"3"^{p'}};
{\ar@{~>}"5";"4"_{q'}};
\endxy
\in \text{Mor}_2(\bold{FinStat}_2)
\]
we associate the non-negative extended real number $\text{RE}_2(\spadesuit)\in [0,\infty]$ given by
\be\label{REDEFX77}
\text{RE}_2(\spadesuit)=\text{RE}(\mu,p,s)+\text{CE}(\spadesuit),
\ee
which we refer to as the \define{2-relative entropy} of $\spadesuit$. We note that the quantity $\text{RE}(\mu,p,s)$ appearing on the RHS of \eqref{REDEFX77} is the relative entropy associated with morphism $(\mu,p,s)$ in $\FinStat$, so that
\[
\text{RE}(\mu,p,s)=D(p,s\circ \mu\circ p),
\]
where $D(-,-)$ is the standard relative entropy. 
\ed

\begin{proposition}\label{RECL979}
2-Relative entropy is convex linear, i.e., if $(X,p)$ is a finite probability space and $\spadesuit^{x}$ is a collection of 2-morphisms in $\bold{FinStat}_2$ indexed by $X$, then
\[
\emph{RE}_2\left(\bigoplus_{x\in X}p_x\spadesuit^x\right)=\sum_{x\in X}p_x \emph{RE}_2(\spadesuit^x).
\]
\end{proposition}
\begin{proof}
Suppose $\bigoplus_{x\in X}p_x\spadesuit^x$ is summarized by the following diagram
\[
\xy0;/r.25pc/:
(0,7.5)*+{\star}="1";
(-20,-7.5)*+{\coprod_{x\in X}U_x}="2";
(20,-7.5)*+{\coprod_{x\in X}V_x}="0";
(-20,-27.5)*+{\coprod_{x\in X}U'_x}="3";
(20,-27.5)*+{\coprod_{x\in X}V'_x}="4";
(0,-42.5)*+{\star}="5";
{\ar@{~>}"1";"2"_{\bigoplus_{x\in X} p_xq^x}};
{\ar@{~>}"1";"0"^{\bigoplus_{x\in X} p_xr^x}};
{\ar@{~>}"2";"0"_{\bigoplus_{x\in X} f_x}};
{\ar@{~>}"3";"4"_{\bigoplus_{x\in X} f'_x}};
{\ar"2";"3"^{\bigoplus_{x\in X} \mu^x}};
{\ar"0";"4"_{\bigoplus_{x\in X} \nu^x}};
{\ar@{~>}@/_1.5pc/"4";"0"_{\bigoplus_{x\in X} t^x.}};
{\ar@{~>}@/^1.5pc/"3";"2"^{\bigoplus_{x\in X} s^x}};
{\ar@{~>}"5";"3"^{{\bigoplus_{x\in X} p_x{q'}^x}}};
{\ar@{~>}"5";"4"_{{\bigoplus_{x\in X} p_x{r'}^x}}};
\endxy
\]
By Theorem~\ref{BFRE} we know that the relative entropy RE is convex linear over 1-morphisms in $\FinStat_2$, and by Propoosition~\ref{CECL979} we know conditional relative entropy is convex linear over 2-morphisms in $\FinStat_2$, thus
\begin{equation}\label{TC771}
\text{RE}\left(\bigoplus_{x\in X}p_x(\mu^x,q^x,s^x)\right)=\sum_{x\in X}p_x\text{RE}(\mu^x,q^x,s^x) \quad \& \quad \text{CE}\left(\bigoplus_{x\in X}p_x\spadesuit^x\right)=\sum_{x\in X}p_x\text{CE}(\spadesuit^x).
\end{equation}
We then have
\begin{eqnarray*}
\text{RE}_2\left(\bigoplus_{x\in X}p_x\spadesuit^x\right)&=&\text{RE}\left(\bigoplus_{x\in X}p_x(\mu^x,q^x,s^x)\right)+\text{CE}\left(\bigoplus_{x\in X}p_x\spadesuit^x\right) \\
&\overset{\eqref{TC771}}=&\sum_{x\in X}p_x\text{RE}(\mu^x,q^x,s^x)+\sum_{x\in X}p_x\text{CE}(\spadesuit^x) \\
&=&\sum_{x\in X}p_x\left(\text{RE}(\mu^x,q^x,s^x)+\text{CE}\left(\spadesuit^x\right)\right) \\
&=&\sum_{x\in X}p_x\text{RE}_2(\spadesuit^x), \\
\end{eqnarray*}
as desired.
\end{proof}

\begin{theorem}\label{REFV19}
Relative entropy is functorial with respect to vertical composition, i.e., if $\clubsuit\circ \spadesuit$ is a vertical composition in $\bold{FinStat}_2$, then $\emph{RE}_2(\clubsuit\circ \spadesuit)=\emph{RE}_2(\clubsuit)+\emph{RE}_2(\spadesuit)$.
\end{theorem}
\begin{proof}
Suppose $\spadesuit$ and $\clubsuit$ are such that the vertical composition $\clubsuit\circ \spadesuit$ is summarized by the following diagram.
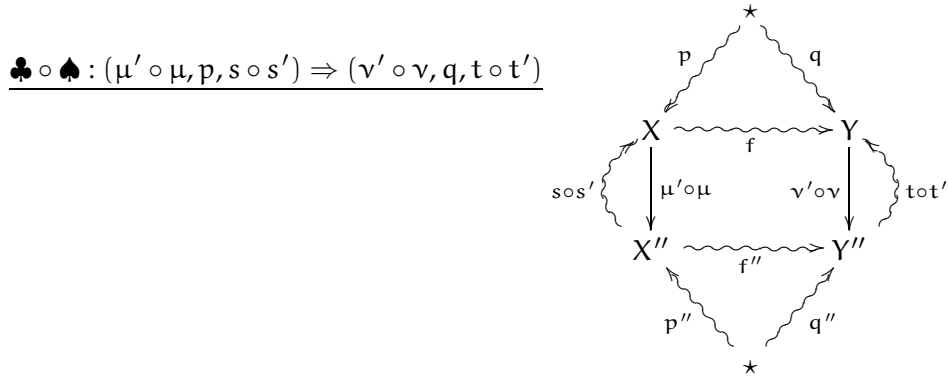
\[
\underline{\clubsuit\circ \spadesuit:(\mu'\circ\mu,p,s\circ s')\Rightarrow (\nu'\circ \nu,q,t\circ t')}
\xy0;/r.25pc/:
(0,7.5)*+{\star}="1";
(-12.5,-7.5)*+{X}="2";
(12.5,-7.5)*+{Y}="0";
(-12.5,-22.5)*+{X''}="3";
(12.5,-22.5)*+{Y''}="4";
(0,-37.5)*+{\star}="5";
{\ar@{~>}"1";"2"_{p}};
{\ar@{~>}"1";"0"^{q}};
{\ar@{~>}"2";"0"_{f}};
{\ar@{~>}"3";"4"_{f''}};
{\ar"2";"3"^{\mu'\circ \mu}};
{\ar"0";"4"_{\nu'\circ \nu}};
{\ar@{~>}@/_1.5pc/"4";"0"_{t\circ t'}};
{\ar@{~>}@/^1.5pc/"3";"2"^{s\circ s'}};
{\ar@{~>}"5";"3"^{p''}};
{\ar@{~>}"5";"4"_{q''}};
\endxy
\]
Then
\begin{eqnarray*}
\text{RE}_2(\clubsuit\circ \spadesuit)&=&\text{RE}(\mu'\circ \mu,p,s\circ s')+\text{CE}(\clubsuit\circ \spadesuit) \\
&=&\text{RE}(\mu,p,s)+\text{RE}(\mu',p',s')+\text{CE}(\clubsuit)+\text{CE}(\spadesuit) \\
&=&\text{RE}_2(\clubsuit)+\text{RE}_2(\spadesuit),
\end{eqnarray*}
where the second equality follows from Theorem~\ref{BFRE} and Theorem~\ref{CEFV19}. 
\end{proof}

\bn
Let $\spadesuit:(\mu,p,s) \Rightarrow (\nu,q,t)$ be a 2-morphism in $\FinStat$, and suppose $s$ and $t$ are optimal hypotheses for $(\mu|p)$ and $(\nu|q)$ as (defined in Definition~\ref{OH19}). Then $\emph{RE}_2(\spadesuit)=0$.
\en

\bprf
Since $s$ and $t$ are optimal hypotheses it follows that $\text{RE}(\mu,p,s)=\text{CE}(\spadesuit)=0$, from which the proposition follows. 
\eprf

\bn
The 2-relative entropy $\emph{RE}_2$ is lower semicontinuous. 
\en

\bprf
Since the 2-relative entropy $\text{RE}_2$ is a linear combination of 1-level relative entropies, and 1-level relative entropies are lower semicontinuous by Theorem~\ref{BFRE}, it follows that $\text{RE}_2$ is lower semicontinuous.
\eprf

\section{Conclusions}

In this work we have constructed a 2-categorical extension $\text{RE}_2$ of the relative entropy functor RE of Baez and Fritz \cite{BFRE}, yielding a new measure of information which we view as a relative measure of information between noisy channels. Moreover, we show that our construction satisfies natural 2-level analogues of functoriality, convex linearity, vanishing under optimal hypotheses and lower semicontinuity.  As the relative entropy functor of Baez and Fritz is uniquely characterized by such properties, it is only natural to question if our 2-level extension $\text{RE}_2$ of RE is also uniquely characterized by the 2-level analogues of such properties. We leave such investigation to future work.

\addcontentsline{toc}{section}{\numberline{}Bibliography}
\bibliographystyle{plain}
\bibliography{2RE_ENT}

\end{document}